\documentclass[letterpaper, 10 pt,journal]{IEEEtran}

\IEEEoverridecommandlockouts                               

\usepackage{graphics} 
\usepackage{graphicx}
\usepackage{epsfig} 
\usepackage{amsmath} 
\usepackage{amssymb}  

\usepackage{tikz} 
\usepackage{pgfplots} 

\usepackage{algpseudocode}
\usepackage{algorithm}
\usepackage{epstopdf}

\usepackage{color}

\usepackage{lipsum}
\usepackage{verbatim}

\usepackage{amsthm}

\usepackage{mathtools}

\newtheorem{theorem}{\bf Theorem} \newtheorem{definition}[theorem]{\bf Definition} 
 \newtheorem{remark}[theorem]{\bf Remark}
  \newtheorem{proposition}[theorem]{\bf Proposition} 
 \newtheorem{example}[theorem]{\bf Example} 
\newtheorem{Algorithm}[theorem]{\bf Algorithm}

\title{\LARGE \bf
A trajectory-based framework for data-driven\\system analysis and control}

\author{Julian Berberich, Frank Allg\"ower
\thanks{The authors are with the Institute for Systems Theory and Automatic Control, University of Stuttgart, 70550 Stuttgart, Germany. 
E-mail:$\{$ julian.berberich, frank.allgower\}@ist.uni-stuttgart.de,
Phone:
+49 711 685-\{67747, 67733\}, Contact for correspondence:
J. Berberich.
The authors thank the German Research Foundation (DFG) for support of this work within the German Excellence Strategy under grant EXC 2075, along with the Max Planck Research School (IMPRS) for Intelligent Systems for their support.
}
}

\begin{document}
\IEEEpubid{\begin{minipage}{\textwidth}\ \\[12pt] This version has been accepted for publication in Proc. European Control Conference (ECC), 2020. Personal use of this material is permitted. Permission
from EUCA must be obtained for all other uses, in any current or future media, including reprinting/republishing this material for advertising or promotional
purposes, creating new collective works, for resale or redistribution to servers or lists, or reuse of any copyrighted component of this work in other works.\end{minipage}}
\maketitle


\begin{abstract}
The vector space of all input-output trajectories of a discrete-time linear time-invariant (LTI) system is spanned by time-shifts of a single measured trajectory, given that the respective input signal is persistently exciting.
This fact, which was proven in the behavioral control framework, shows that a single measured trajectory can capture the full behavior of an LTI system and might therefore be used directly for system analysis and controller design, without explicitly identifying a model.
In this paper, we translate the result from the behavioral context to the classical state-space control framework and we extend it to certain classes of nonlinear systems, which are linear in suitable input-output coordinates.
Moreover, we show how this extension can be applied to the data-driven simulation problem, where we introduce kernel-methods to obtain a rich set of basis functions.
\end{abstract}

\section{Introduction}
Finding rigorous and efficient ways to integrate data into control theory has been a problem of great interest for many decades.
Since most of the classical contributions in control theory rely on model knowledge, the problem of finding such a model from measured data, i.e., system identification, has become a mature research field~\cite{Ljung87}.
More recently, learning controllers directly from data has received increasing interest, not least due to many successful practical applications of reinforcement learning techniques~\cite{Sutton98}.
However, as is thoroughly evaluated in~\cite{Recht18}, such methods typically require large amounts of data, they are often not reproducible, and their analysis rarely addresses rigorous guarantees on, e.g., stability of the closed loop.
Also in the control community, several approaches for the direct design of controllers from data have been proposed.
Established methods include the Virtual Reference Feedback Tuning paradigm~\cite{Campi02} or Iterative Feedback Tuning~\cite{Hjalmarsson98}.
However, fundamental problems such as the direct data-driven design of linear quadratic optimal controllers with guarantees from finite noisy data have only been considered recently~\cite{Dean19,Umenberger18}.

\IEEEpubidadjcol
In this paper, we consider an alternative, unitary framework for data-driven control theory, which allows for the development of various system analysis and controller design methods based directly on measured data.
This framework relies on the characterization of all trajectories of an unknown system using a single measured data trajectory.
The latter problem has been solved in the context of behavioral systems theory for discrete-time linear time-invariant (LTI) systems in~\cite{Willems05}.
In the behavioral approach, a system is not defined via a differential or difference equation with inputs and outputs, but rather as the space of all system trajectories~\cite{Polderman98,Willems91}.
Thus, it is naturally well-suited for the development of purely data-driven approaches to system analysis and control.

Recently, there have been various contributions, which use the result of~\cite{Willems05} for direct data-driven system analysis and control.
In~\cite{Yang15,Coulson19}, a data-driven MPC scheme relying on~\cite{Willems05} is suggested to control unknown systems.
A stochastic analysis of this scheme and an application to power systems are detailed in~\cite{Coulson19b} and~\cite{Huang19}, respectively.
Moreover,~\cite{Berberich19b} provides a first theoretical analysis of stability and robustness of a data-driven MPC scheme based on terminal equality constraints.
In~\cite{Persis19}, a data-driven closed-loop parametrization under state-feedback is derived and employed to design stabilizing and LQR controllers.
This approach is extended to robust design from noisy data in~\cite{Berberich19c}.
Further,~\cite{Waarde19} provides a general framework for analyzing data-driven problems with not persistently exciting data.
Finally, data-based conditions for dissipativity are suggested in~\cite{Romer19}.
Altogether, this indicates a great potential of the work of~\cite{Willems05} for direct data-driven analysis and control.
In this paper, we consider the work of~\cite{Willems05} in the classical control framework and extend it to certain classes of nonlinear systems.
Moreover, we illustrate the usefulness of this extension via a novel kernel-based approach to nonlinear data-driven simulation.

The remainder of this paper is structured as follows.
In Section~\ref{sec:linear_systems}, we phrase the main theorem of~\cite{Willems05}, which uses measured data to characterize all system trajectories, in the classical control setting, and we show how this result can be improved by weaving multiple such trajectories together.
In Section~\ref{sec:nonlinear_systems}, we provide a novel extension of~\cite{Willems05} to classes of nonlinear systems, which are linear in suitably chosen and known nonlinear coordinates.
Building on these results, we solve the data-driven simulation problem for such nonlinear systems in Section~\ref{sec:DD_sim}.
The paper is concluded in Section~\ref{sec:conclusion}.

\section{Setting}

We denote the set of integers in the interval $[a,b]$ by $\mathbb{I}_{[a,b]}$.
The Kronecker product is written as $\otimes$.
For a sequence $\{x_k\}_{k=0}^{N-1}$, we define the Hankel matrix\newpage
\begin{align*}
H_L&(x)\coloneqq\begin{bmatrix}x_0 & x_1& \dots & x_{N-L}\\
x_1 & x_2 & \dots & x_{N-L+1}\\
\vdots & \vdots & \ddots & \vdots\\
x_{L-1} & x_{L} & \dots & x_{N-1}
\end{bmatrix}.
\end{align*}
For a stacked window of the sequence, we write
\begin{align*}
x_{[a,b]}=\begin{bmatrix}x_{a}\\\vdots\\x_{b}\end{bmatrix}.
\end{align*}
Further, $x$ will denote either the sequence itself or the stacked vector $x_{[0,N-1]}$ containing all of its components.
A key assumption for our results will be persistence of excitation of the input signal, as captured in the following standard definition.

\begin{definition}
We say that a signal $\{x_k\}_{k=0}^{N-1}$ with $x_k\in\mathbb{R}^n$ is persistently exciting of order $L$ if $\text{rank}(H_L(x))=nL$.
\end{definition}

Note that the above definition implies $N\geq(n+1)L-1$.
This means that, for a signal to be persistently exciting, it is not sufficient that its time-shifts are linearly independent, but the signal must also be long enough.
A large part of this paper deals with discrete-time multi-input multi-output LTI systems of the form
\begin{align}\label{eq:LTI_sys}
\begin{split}
x_{k+1}&=Ax_k+Bu_k,\>\> x_0=\bar{x},\\
y_k&=Cx_k+Du_k,
\end{split}
\end{align}
where the matrices $A,B,C,D$ as well as the initial condition $\bar{x}$ are unknown and only input-output data $\{u_k,y_k\}_{k=0}^{N-1}$, which may be obtained from~\eqref{eq:LTI_sys} via simulation or an experiment, is available.
Throughout this paper, $n$ denotes the order of the unknown system which is only assumed to be known in terms of a potentially rough upper bound.
Further, we denote the input and output dimension by $m$ and $p$, respectively.

We will use a single trajectory to characterize all other trajectories, which might be produced from the system~\eqref{eq:LTI_sys}, i.e., which satisfy the following definition.

\begin{definition}\label{def:trajectory_of}
We say that an input-output sequence $\{u_k,y_k\}_{k=0}^{N-1}$ is a trajectory of an LTI system $G$, if there exists an initial condition $\bar{x}\in\mathbb{R}^n$ as well as a state sequence $\{x_k\}_{k=0}^{N}$ such that
\begin{align*}
x_{k+1}&=Ax_k+Bu_k,\>\>x_0=\bar{x}\\
y_k&=Cx_k+Du_k,
\end{align*}
for $k=0,\dots,N-1$, where $(A,B,C,D)$ is a minimal realization of $G$.
\end{definition}

It follows from linearity that the set of all trajectories of an LTI system in the sense of Definition~\ref{def:trajectory_of} is a vector space.
As we will see in Section~\ref{sec:linear_systems}, a basis for this vector space is formed by time-shifts of a single measured trajectory, given that the respective input signal is persistently exciting.

Throughout this paper, we make extensive use of the well-known fact that any LTI system admits a controllable and observable minimal realization.
The particular choice of a specific minimal realization is however not relevant.
Further, using LTI system properties, it is easy to show that any fixed window of an input-output trajectory $\{u_k,y_k\}_{k=a}^{b}$ induces a unique state trajectory $\{x_k\}_{k=a}^{b}$ (in a given minimal realization), whenever $b-a\geq n-1$.

\section{Trajectory-based representation of linear systems} \label{sec:linear_systems}

In this section, we translate the main result of~\cite{Willems05}, which characterizes the trajectory space of an unknown system from measured data, to the classical state-space control framework.
While the behavioral theory is naturally well-suited for such a result, we illustrate that it can also be formulated in the classical framework in an elegant way.
Further, we show how a required persistence of excitation assumption can be relaxed by weaving multiple trajectories together to achieve an overall larger time horizon.

The following result is the correspondence of~\cite[Theorem 1]{Willems05} in the classical control setting and it will serve as the basis for the remainder of this paper.

\begin{theorem}\label{thm:hankel_disc_time}
Suppose $\{u_k,y_k\}_{k=0}^{N-1}$ is a trajectory of an LTI system $G$, where $u$ is persistently exciting of order $L+n$.
Then, $\{\bar{u}_k,\bar{y}_k\}_{k=0}^{L-1}$ is a trajectory of $G$ if and only if there exists $\alpha\in\mathbb{R}^{N-L+1}$ such that
\begin{align}\label{eq:thm_hankel}
\begin{bmatrix}H_L(u)\\H_L(y)\end{bmatrix}\alpha
=\begin{bmatrix}\bar{u}\\\bar{y}\end{bmatrix}.
\end{align}
\end{theorem}
\begin{proof}
This is a direct application of~\cite[Theorem 1]{Willems05} to the special case of controllable state-space systems.
\end{proof}

Note that~\eqref{eq:thm_hankel} is equivalent to
\begin{align}\label{eq:time_shift1}
\bar{u}_{[0,L-1]}&=\sum_{i=0}^{N-L}\alpha_iu_{[i,L-1+i]},\\
\bar{y}_{[0,L-1]}&=\sum_{i=0}^{N-L}\alpha_iy_{[i,L-1+i]},\label{eq:time_shift2}
\end{align}
i.e., the trajectory space is spanned by time-shifts of the measured trajectory.
Similarly, it holds for the state that
\begin{align}
\bar{x}_{[0,L-1]}=\sum_{i=0}^{N-L}\alpha_ix_{[i,L-1+i]},
\end{align}
where $\bar{x}$ and $x$ are states corresponding to $(\bar{u},\bar{y})$ and $(u,y)$, respectively, in the same minimal realization.
The ``if''-direction in Theorem~\ref{thm:hankel_disc_time} follows directly from the fact that $G$ is LTI, without adhering to the persistence of excitation assumption.
The intuition about the ``only if''-direction is sketched in the following.
Take any trajectory $\{\bar{u}_k,\bar{y}_k\}_{k=0}^{L-1}$ of $G$.
Clearly, $L$ degrees of freedom in the input are required to choose $\alpha\in\mathbb{R}^{N-L+1}$ such that~\eqref{eq:time_shift1} holds.
Additional $n$ degrees of freedom can then be used to attain the internal initial condition $\bar{x}_0$.
Since $\{\bar{y}_k\}_{k=0}^{L-1}$ is a linear combination of $\{\bar{u}_k\}_{k=0}^{L-1}$ and $\bar{x}_0$, this is enough to find an $\alpha$ which satisfies both~\eqref{eq:time_shift1} and~\eqref{eq:time_shift2}, and thus~\eqref{eq:thm_hankel}.
Therefore, persistence of excitation of order $L+n$ is required for the equivalence in Theorem~\ref{thm:hankel_disc_time}.

Theorem~\ref{thm:hankel_disc_time} shows that all trajectories of an unknown LTI system can be constructed from a \emph{single} persistently exciting trajectory.
Equivalently, the vector space of all system trajectories is equal to the range of a data-dependent Hankel matrix.
Thus, in a way, the measured input-output trajectory serves as a system representation on its own, without using it explicitly to identify a model.
Prior knowledge of the unknown system's order is only needed implicitly in Theorem~\ref{thm:hankel_disc_time} through the condition that $u$ has to be persistently exciting of order $L+n$.
Hence, if the amount of available data $N$ is significantly larger than $n$ and the input is persistently exciting of a sufficiently high order, a rough upper bound on $n$ suffices to apply Theorem~\ref{thm:hankel_disc_time}.

As described above, persistence of excitation is necessary for the equivalence in Theorem~\ref{thm:hankel_disc_time}.
Note however that it also sets a fundamental limit on the application of Theorem~\ref{thm:hankel_disc_time}:
In order to span the space of all trajectories of length $L$, Theorem~\ref{thm:hankel_disc_time} requires $N\geq(m+1)(L+n)-1$ or, equivalently, $L\leq\frac{N+1}{m+1}-n$.
Loosely speaking, if $m=1$, $L$ can only be half as long as $N$ and, with increasing input dimension $m$, the maximum length $L$ decreases by a factor of $\frac{1}{m+1}$.

An intuitive solution to overcome this limitation would be to weave several, say $\xi\in\mathbb{N}$, trajectories of length $L$ together to construct an overall trajectory of length $\xi L$.
This is however not trivial, since the internal states of the separate trajectories have to align at the intersections.
In~\cite[Lemma 3]{Markovsky05}, it is shown that two distinct input-output trajectories can be weaved together if they align over a sufficiently long window at their intersection.
The following result is an extension of~\cite[Lemma 3]{Markovsky05} to more than two trajectories.

\begin{proposition}\label{prop:hankel_weave}
Suppose $\{u_k,y_k\}_{k=0}^{N-1}$ is a trajectory of an LTI system $G$, where $u$ is persistently exciting of order $L+n$.
Then, $\{\bar{u}_k,\bar{y}_k\}_{k=0}^{\tilde{L}-1}$ with $\tilde{L}=\xi L+(1-\xi)n$, $\xi\in\mathbb{N}$, is a trajectory of $G$ if and only if there exist $\alpha^i\in\mathbb{R}^{N-L+1},i\in\mathbb{I}_{[1,\xi]},$ such that
\begin{align}
\label{eq:prop_hankel_weave1}
\begin{bmatrix}
H_L(u) & 0 \\0 &I_{\xi-1}\otimes H_{L-n}\left(u_{[n,N-1]}\right)\\
H_L(y) & 0 \\0 &I_{\xi-1}\otimes H_{L-n}\left(y_{[n,N-1]}\right)
\end{bmatrix}
\begin{bmatrix}
\alpha^1\\\vdots\\\alpha^\xi
\end{bmatrix}
=&\begin{bmatrix}
\bar{u}_{[0,\tilde{L}-1]}\\\bar{y}_{[0,\tilde{L}-1]}
\end{bmatrix},\\\nonumber\\
\label{eq:prop_hankel_weave2}
H_n\left(u_{[L-n,N-1]}\right)\alpha^i=H_n\left(u_{[0,N-L+n-1]}\right)&\alpha^{i+1},\\
\label{eq:prop_hankel_weave3}
H_n\left(y_{[L-n,N-1]}\right)\alpha^i=H_n\left(y_{[0,N-L+n-1]}\right)&\alpha^{i+1},\\
\nonumber i\in\mathbb{I}_{[1,\xi-1]}.
\end{align}
\end{proposition}
\begin{proof}
\textbf{If.}
Define $\{\bar{u}^i_k,\bar{y}^i_k\}_{k=0}^{L-1}$ via
\begin{align*}
\begin{bmatrix}\bar{u}_{[0,L-1]}^i\\\bar{y}_{[0,L-1]}^i\end{bmatrix}=
\begin{bmatrix}H_L(u)\\H_L(y)\end{bmatrix}\alpha^i,\>\>i\in\mathbb{I}_{[1,\xi]},
\end{align*}
and note that~\eqref{eq:prop_hankel_weave1} means that $\{\bar{u}_k,\bar{y}_k\}_{k=0}^{\tilde{L}-1}$ is a stacked version of the sequences $\{\bar{u}^i_k,\bar{y}^i_k\}_{k=0}^{L-1}$ in the sense that
\begin{align}\label{eq:thm_stacked_proof1}
\bar{u}_{[0,\tilde{L}-1]}=
\begin{bmatrix}
\bar{u}_{[0,L-1]}^1\\ \bar{u}_{[n,L-1]}^2\\\vdots\\\bar{u}_{[n,L-1]}^\xi
\end{bmatrix},\>\>
\bar{y}_{[0,\tilde{L}-1]}=
\begin{bmatrix}
\bar{y}_{[0,L-1]}^1\\ \bar{y}_{[n,L-1]}^2\\\vdots\\\bar{y}_{[n,L-1]}^\xi
\end{bmatrix}.
\end{align}
According to Theorem~\ref{thm:hankel_disc_time}, the sequences $\{\bar{u}^i_k,\bar{y}^i_k\}_{k=0}^{L-1}$ are trajectories of $G$.
Further,~\eqref{eq:prop_hankel_weave2} and~\eqref{eq:prop_hankel_weave3} imply that, at the transitions between the separate trajectories, they align over windows of length $n$, i.e.,
\begin{align}\label{eq:thm_stacked_proof2a}
\bar{u}^i_{[L-n,L-1]}&=\bar{u}^{i+1}_{[0,n-1]},\>\>i\in\mathbb{I}_{[1,\xi-1]},\\\label{eq:thm_stacked_proof2b}
\bar{y}^i_{[L-n,L-1]}&=\bar{y}^{i+1}_{[0,n-1]},\>\>i\in\mathbb{I}_{[1,\xi-1]}.
\end{align}
Denote by $\{\bar{x}^i_k\}_{k=0}^{L-1}$ the state trajectory corresponding to $\{\bar{u}^i_k,\bar{y}^i_k\}_{k=0}^{L-1}$ in some minimal realization of $G$.
The conditions~\eqref{eq:thm_stacked_proof2a} and~\eqref{eq:thm_stacked_proof2b} imply that, at the transitions between the separate trajectories, the internal states align, i.e., $\bar{x}^i_L=\bar{x}^{i+1}_n$, and thus, $\{\bar{u}_k,\bar{y}_k\}_{k=0}^{\tilde{L}-1}$ is a trajectory of $G$.

\textbf{Only If.}
Suppose $\{\bar{u}_k,\bar{y}_k\}_{k=0}^{\tilde{L}-1}$ is a trajectory of $G$.
Define $\{\bar{u}^i_k,\bar{y}^i_k\}_{k=0}^{L-1}$, $i\in\mathbb{I}_{[1,\xi]}$, according to~\eqref{eq:thm_stacked_proof1}-\eqref{eq:thm_stacked_proof2b} and note that any of these sequences is itself a trajectory of $G$.
Hence, it follows directly from Theorem~\ref{thm:hankel_disc_time} that there exist $\alpha^i\in\mathbb{R}^{N-L+1}$, $i\in\mathbb{I}_{[1,\xi]},$ such that~\eqref{eq:prop_hankel_weave1}-\eqref{eq:prop_hankel_weave3} hold.
\end{proof}

Proposition~\ref{prop:hankel_weave} weaves multiple trajectories $\{\bar{u}^i_k,\bar{y}^i_k\}_{k=0}^{L-1}$ together to form a single, longer sequence $\{\bar{u}_k,\bar{y}_k\}_{k=0}^{\tilde{L}-1}$.
To make this sequence a trajectory of $G$, it only needs to be ensured that the shorter trajectories align over at least $n$ steps at their intersections.
Note that the number of trajectories $\xi$ can be chosen arbitrarily large and thus, Proposition~\ref{prop:hankel_weave} can be used to construct trajectories of arbitrary length, using a single measured trajectory of \emph{finite} length.
Although we assume for notational simplicity that all trajectories contributing to the overall trajectory are of equal length, the same idea can be applied to weave trajectories of different lengths together.
Further, one can straightforwardly employ measurements from multiple experiments of possibly different time horizons.

\section{Trajectory-based representation of nonlinear systems} \label{sec:nonlinear_systems}

In this section, we extend Theorem~\ref{thm:hankel_disc_time} to certain classes of nonlinear systems.
In particular, we consider the special cases of Hammerstein and Wiener systems.
More generally, this allows us to extend Theorem~\ref{thm:hankel_disc_time} to all systems, which are linear in suitably chosen and known input-output coordinates.
During the last decades, there have been many contributions to identify Hammerstein and Wiener systems from data~\cite{Bai98,Westwick96}.
Our results can be seen as an alternative to the identification of such systems, using a single measured trajectory to represent them.

\subsection{Hammerstein systems}\label{sec:Hammerstein}

A Hammerstein system is a nonlinear system, composed of a static nonlinearity followed by an LTI system, i.e.,
\begin{align}\label{eq:Hammerstein_sys}
\begin{split}
x_{k+1}&=Ax_k+B\psi(u_k),\>\>x_0=\bar{x},\\
y_k&=Cx_k+D\psi(u_k),
\end{split}
\end{align}
with a nonlinear function $\psi:\mathbb{R}^m\to\mathbb{R}^{\tilde{m}}$.
In the following, we deal only with the case $\tilde{m}=1$ for notational simplicity, but the same ideas can be employed for $\tilde{m}>1$.
We assume that $\psi$ can be written as $\psi(u)=\sum_{i=1}^ra_i\psi_i(u)$, with $a_i$ not all zero, for $r$ known basis functions $\psi_i$.
Further, we define the auxiliary input trajectory $\{v_k\}_{k=0}^{N-1}$ with components
\begin{align}\label{eq:aux_input}
v_k=\begin{bmatrix}\psi_1(u_k)\\\vdots\\\psi_r(u_k)\end{bmatrix}.
\end{align}
The following result uses the fact that~\eqref{eq:Hammerstein_Wiener_sys} can also be viewed as a linear map from $v$ to $y$.

\begin{proposition}\label{prop:hankel_Hammerstein}
Suppose $\{u_k,y_k\}_{k=0}^{N-1}$ is a trajectory of a Hammerstein system~\eqref{eq:Hammerstein_sys}, where $v$ from~\eqref{eq:aux_input} is persistently exciting of order $L+n$.
Then, $\{\bar{u}_k,\bar{y}_k\}_{k=0}^{L-1}$ is a trajectory of~\eqref{eq:Hammerstein_sys} if and only if there exists $\alpha\in\mathbb{R}^{N-L+1}$ such that
\begin{align}\label{eq:prop_hankel_Hammerstein}
\begin{bmatrix}H_L(v)\\H_L(y)\end{bmatrix}\alpha
=\begin{bmatrix}\bar{v}\\\bar{y}\end{bmatrix},
\end{align}
where $\{\bar{v}_k\}_{k=0}^{L-1}$ is the sequence with components
\begin{align*}
\bar{v}_k=\begin{bmatrix}\psi_1(\bar{u}_k)\\\vdots\\\psi_r(\bar{u}_k)\end{bmatrix}.
\end{align*}
\end{proposition}
\begin{proof}
Define the LTI system
\begin{align}\label{eq:Hammerstein_sys_lifted}
\begin{split}
x_{k+1}&=Ax_k+\tilde{B} v_k,\\
y_k&=Cx_k+\tilde{D} v_k,
\end{split}
\end{align}
with input $v$ and output $y$, with $a^\top=\begin{bmatrix}a_1 & \dots & a_r\end{bmatrix}$, $\tilde{B}=Ba^\top,\tilde{D}=Da^\top$, and $A,B,C,D$ from~\eqref{eq:Hammerstein_sys}.
Clearly, a sequence $\{\bar{u}_k,\bar{y}_k\}_{k=0}^{L-1}$ is a trajectory of~\eqref{eq:Hammerstein_sys} if and only if $\{\bar{v}_k,\bar{y}_k\}_{k=0}^{L-1}$ is a trajectory of~\eqref{eq:Hammerstein_sys_lifted}.
Further, using that $v$ is persistently exciting and~\eqref{eq:Hammerstein_sys_lifted} is controllable since the $a_i$'s are not all zero, Theorem~\ref{thm:hankel_disc_time} implies that $\{\bar{v}_k,\bar{y}_k\}_{k=0}^{L-1}$ is a trajectory of~\eqref{eq:Hammerstein_sys_lifted} if and only if there exists $\alpha\in\mathbb{R}^{N-L+1}$ such that~\eqref{eq:prop_hankel_Hammerstein} holds, which was to be shown.
\end{proof}

For the application of Proposition~\ref{prop:hankel_Hammerstein}, the basis functions $\psi_i$ of $\psi$ have to be known.
In practice, it may be adequate to simply choose sufficiently many basis functions, thereby approximating the true ones.
Note however that the number of basis functions $r$ enters into the persistence of excitation assumption on the auxiliary input $v$.
To be more precise, for $v$ to be persistently exciting of order $L+n$, it is necessary that $N\geq(r+1)(L+n)-1$ and hence, Proposition~\ref{prop:hankel_Hammerstein} does not allow for arbitrarily many basis functions.
Nevertheless, we show in Section~\ref{sec:DD_sim} that, for the data-driven simulation problem, Proposition~\ref{prop:hankel_Hammerstein} leads to meaningful results even if infinitely many basis functions are chosen implicitly via a kernel function.

\subsection{Wiener systems}\label{sec:Wiener}

A Wiener system consists of an LTI system followed by a static nonlinearity, i.e., it is of the form
\begin{align}\label{eq:Wiener_sys}
\begin{split}
x_{k+1}&=Ax_k+Bu_k,\>\>x_0=\bar{x},\\
y_k&=\phi(Cx_k+Du_k),
\end{split}
\end{align}
with a nonlinear function $\phi:\mathbb{R}^{\tilde{p}}\to\mathbb{R}^p$.
Similar to Section~\ref{sec:Hammerstein}, we consider in the following only the case $\tilde{p}=1$.
To apply the same reasoning as for Hammerstein systems, we assume that $\phi$ is invertible and that its inverse admits a basis function decomposition as $\phi^{-1}(y)=\sum_{i=1}^q b_i\tilde{\phi}_i(y)$ with $q$ known basis functions $\tilde{\phi}_i$.
We define an auxiliary output trajectory $\{z_k\}_{k=0}^{N-1}$ with components
\begin{align}\label{eq:aux_output}
z_k=\begin{bmatrix}\tilde{\phi}_1(y_k)\\\vdots\\\tilde{\phi}_q(y_k)\end{bmatrix},
\end{align}
which will serve as the output of an equivalent LTI system.
The following result is the correspondence of Proposition~\ref{prop:hankel_Hammerstein} for the Wiener system case.

\begin{proposition}\label{prop:hankel_Wiener}
Suppose $\{u_k,y_k\}_{k=0}^{N-1}$ is a trajectory of a Wiener system~\eqref{eq:Wiener_sys}, where $u$ is persistently exciting of order $L+n$.
Then, $\{\bar{u}_k,\bar{y}_k\}_{k=0}^{L-1}$ is a trajectory of~\eqref{eq:Wiener_sys} if there exists $\alpha\in\mathbb{R}^{N-L+1}$ such that
\begin{align}\label{eq:prop_hankel_Wiener}
\begin{bmatrix}H_L(u)\\H_L(z)\end{bmatrix}\alpha
=\begin{bmatrix}\bar{u}_{[0,L-1]}\\\bar{z}_{[0,L-1]}\end{bmatrix},
\end{align}
where $\{\bar{z}_k\}_{k=0}^{L-1}$ is the sequence with components
\begin{align*}
\bar{z}_k=\begin{bmatrix}\tilde{\phi}_1(\bar{y}_k)\\\vdots\\\tilde{\phi}_q(\bar{y}_k)\end{bmatrix}.
\end{align*}
\end{proposition}
\begin{proof}
This can be shown using similar arguments as in the proof of Proposition~\ref{prop:hankel_Hammerstein}. Therefore, the proof is omitted.
\end{proof}

Contrary to the Hammerstein case, the above result does not pose any limit on the maximal number of basis functions we may choose.
However, they represent the \emph{inverse} of $\phi$ and are thus more difficult to select in applications.
Further, the ``only if''-direction does in general not hold for Wiener systems since the map $u\mapsto z$ is not necessarily linear.

\begin{remark}
From the perspective of Koopman operator theory, there has recently been a renewed interest in viewing nonlinear systems as linear systems in lifted state coordinates~\cite{Budisic12}.
In a similar fashion, Propositions~\ref{prop:hankel_Hammerstein} and~\ref{prop:hankel_Wiener} can be combined directly to provide trajectory-based representations of nonlinear systems, which are linear in suitable higher-dimensional input-output coordinates.
Even if such coordinates do not exist or are not known, one may in practice simply choose sufficiently many basis functions to approximate the unknown nonlinear system.
In Section~\ref{sec:DD_sim}, we illustrate the effectiveness of this approach for the data-driven simulation problem.
Note that considering systems which are linear in suitable input-output coordinates is more restrictive than dealing with systems which are linear in certain lifted state coordinates.
On the other hand, in contrast to many methods related to Koopman operator theory, the present setting does not require state measurements, but only input-output data.
\end{remark}

\section{Data-driven simulation}\label{sec:DD_sim}

The data-driven simulation problem is concerned with the computation of an unknown system's output resulting from the application of a given input, using no model but only a previously measured input-output trajectory.
Its solution is described in the behavioral context in~\cite{Markovsky08}.
Loosely speaking, the idea is to fix $\bar{u}$ in~\eqref{eq:thm_hankel} to first solve $\bar{u}=H_L(u)\alpha$ for $\alpha$, in order to then compute the new predicted output $\bar{y}=H_L(y)\alpha$.
To fix a unique such output, initial conditions have to be specified~\cite{Markovsky08}.
Since a state-space model is not available, we consider an \emph{initial input-output trajectory} over a length of at least $n$, since this induces a unique initial state in \emph{some} minimal realization.
The following is the main result of~\cite{Markovsky08}.

\begin{proposition}\label{prop:DD_sim}
Suppose $\{u_k,y_k\}_{k=0}^{N-1}$ is a trajectory of a discrete-time LTI system $G$, where $u$ is persistently exciting of order $L+n$.
Let $\{\bar{u}_k,\bar{y}_k\}_{k=0}^{L-1}$ be an arbitrary trajectory of $G$.
If $\nu\geq n$, then there exists an $\alpha\in\mathbb{R}^{N-L+1}$ to
\begin{align}\label{eq:prop_DD_sim}
\begin{bmatrix}H_L(u)\\H_\nu\left(y_{[0,N-L+\nu-1]}\right)\end{bmatrix}
\alpha=\begin{bmatrix}\bar{u}\\\bar{y}_{[0,\nu-1]}\end{bmatrix}.
\end{align}
Further, it holds that $\bar{y}=H_L(y)\alpha$.
\end{proposition}
\begin{proof}
This follows directly from the corresponding result in the behavioral framework~\cite[Proposition 1]{Markovsky08}.
\end{proof}

The main idea of Proposition~\ref{prop:DD_sim} is that the input $\{\bar{u}_k\}_{k=\nu}^{L-1}$ together with the initial trajectory $\{\bar{u}_k,\bar{y}_k\}_{k=0}^{\nu-1}$ fixes a vector $\alpha$ which can be used to uniquely predict the remaining elements of $\bar{y}$.
The condition $\nu\geq n$ means that $\nu$ is an upper bound on the order $n$ of $G$ and implies that $\{\bar{u}_k,\bar{y}_k\}_{k=0}^{\nu-1}$ specifies a unique initial condition for the internal state.

\begin{algorithm}
\begin{Algorithm}\label{alg:DD_sim}
\normalfont{\textbf{Data-driven simulation}}\\
\textbf{Given:} Data $\{u_k,y_k\}_{k=0}^{N-1}$, a new input $\{\bar{u}_k\}_{k=\nu}^{L-1}$, initial conditions $\{\bar{u}_k,\bar{y}_k\}_{k=0}^{\nu-1}$.
\begin{enumerate}
\item Solve~\eqref{eq:prop_DD_sim} for $\alpha$.
\item Compute the remaining elements of $\bar{y}$ as $\bar{y}=H_L(y)\alpha$.
\end{enumerate}
\end{Algorithm}
\end{algorithm}

The practical application of Proposition~\ref{prop:DD_sim} is illustrated in Algorithm~\ref{alg:DD_sim}.
Although the classical simulation problem is commonly approached using a model, it can be solved in the proposed trajectory-based framework using a single measured input-output trajectory.
Several extensions of Proposition~\ref{prop:DD_sim} have been suggested to account for noise~\cite{Markovsky06}, to simulate systems in closed loop~\cite{Markovsky10}, and to find feedforward controllers~\cite{Markovsky08}, but nonlinear systems have not been addressed in the literature.
Due to noise or numerical inaccuracies, the system of equations~\eqref{eq:prop_DD_sim} can usually not be solved exactly.
Instead, $\alpha$ can be computed via a simple least-squares optimization problem.
Denote
\begin{align*}
H_{L,\nu}(u,y)\coloneqq\begin{bmatrix}H_L(u)\\H_\nu\left(y_{[0,N-L+\nu-1]}\right)\end{bmatrix},\>\>\bar{w}\coloneqq\begin{bmatrix}\bar{u}\\\bar{y}_{[0,\nu-1]}\end{bmatrix}.
\end{align*}
In practice, the system of equations~\eqref{eq:prop_DD_sim} can be replaced by
\begin{align}\label{eq:DD_sim_lsq}
\underset{\alpha\in\mathbb{R}^{N-L+1}}{\text{minimize}}\left\lVert H_{L,\nu}(u,y)\alpha-\bar{w}\right\rVert_2^2.
\end{align}
In case that $H_{L,\nu}(u,y)$ contains noisy data, a solution $\alpha$ with small norm reduces the influence of the noise on the simulation accuracy.
Therefore, it is desirable to penalize the norm of $\alpha$, leading to the regularized least-squares problem
\begin{align}\label{eq:DD_sim_lsq_reg}
\underset{\alpha\in\mathbb{R}^{N-L+1}}{\text{minimize}}
\left\lVert H_{L,\nu}(u,y)\alpha-\bar{w}\right\rVert_2^2
+\lambda\lVert\alpha\rVert_2^2,
\end{align}
where $\lambda>0$ is a regularization parameter.
As an alternative, one may consider general quadratic regularization terms $\lVert \alpha\rVert_P^2$ with $P\succ0$ or an $\ell_1$-regularization.
In the following, we show how kernel methods can be employed to derive an appealing reformulation of Problem~\eqref{eq:DD_sim_lsq_reg} for the class of nonlinear systems considered in Section~\ref{sec:nonlinear_systems}.

Let us consider a Hammerstein-Wiener system of the form
\begin{align}\label{eq:Hammerstein_Wiener_sys}
\begin{split}
x_{k+1}&=Ax_k+B\psi(u_k),\>\>x_0=\bar{x},\\
y_k&=\Phi(Cx_k+D\psi(u_k)).
\end{split}
\end{align}
According to Propositions~\ref{prop:hankel_Hammerstein} and~\ref{prop:hankel_Wiener}, the trajectory space of~\eqref{eq:Hammerstein_Wiener_sys} is spanned by Hankel matrices containing data in the lifted coordinates $v$ and $z$, as defined in~\eqref{eq:aux_input} and~\eqref{eq:aux_output}.
Thus, for the system class~\eqref{eq:Hammerstein_Wiener_sys}, the optimization problem~\eqref{eq:DD_sim_lsq_reg} takes the form
\begin{align}\label{eq:DD_sim_lsq_reg_lifted}
\underset{\alpha\in\mathbb{R}^{N-L+1}}{\text{minimize}}
\left\lVert H_{L,\nu}(v,z)\alpha-\begin{bmatrix}\bar{v}\\\bar{z}_{[0,\nu-1]}
\end{bmatrix}\right\rVert_2^2
+\lambda\lVert\alpha\rVert_2^2.
\end{align}
In the following, we write $\psi^r(u_k)$ and $\tilde{\Phi}^q(y_k)$ for the stacked inputs $v_k$ and outputs $z_k$ at time $k$, respectively.
Note that Problem~\eqref{eq:DD_sim_lsq_reg_lifted} does not depend explicitly on these vectors, but only on their scalar product.
This allows for an application of the kernel trick, which can be used to compute such inner products implicitly~\cite{Burges98}.
Define kernel functions as
\begin{align*}
\mathcal{K}_\psi(u_k^1,u_k^2)&=\psi^r(u_k^1)^\top\psi^r(u_k^2),\\
\mathcal{K}_\Phi(y_k^1,y_k^2)&=\tilde{\Phi}^q(y_k^1)^\top\tilde{\Phi}^q(y_k^2),
\end{align*}
and note that~\eqref{eq:DD_sim_lsq_reg_lifted} depends only on those kernels, but not explicitly on the basis functions $\psi^r,\tilde{\Phi}^q$.
Thus, for the implementation, it suffices to select a kernel, which then implicitly implies a set of basis functions for the nonlinearities $\psi$ and $\tilde{\Phi}$.
For instance, if $m=1$, a squared exponential kernel of the form 
\begin{align}\label{eq:kernel}
\begin{split}
\mathcal{K}_\psi(u_1,u_2)=e^{-\frac{(u_1-u_2)^2}{2\sigma^2}},
\end{split}
\end{align}
for some hyperparameter $\sigma>0$, corresponds to an infinite set of basis functions.
If the set of basis functions corresponding to the chosen kernel contains all basis functions of $\psi$ and $\tilde{\Phi}$, then the data-driven simulation problem can be solved exactly for the considered class of nonlinear systems.
In fact, as we will see in the following example, the data-driven simulation problem can be solved accurately, even if the data is affected by noise and the true basis functions are only represented approximately by the chosen kernel.

\begin{example}\label{ex:example}
We consider a Hammerstein system~\eqref{eq:Hammerstein_sys} with nonlinearity $\psi(u)=\sin(u)$ and the system matrices
\begin{align*}
A&=\begin{bmatrix}0.4&-0.3&0&0.1\\-0.3&0&0.8&-0.1\\
0.1&-0.7&-0.4&0\\0.2&-0.5&0.5&0.4\end{bmatrix},\>\>
B=\begin{bmatrix}0\\-1\\1.4\\0\end{bmatrix},\\
C&=\begin{bmatrix}-0.7&0&-2&0.4\end{bmatrix},\>\>D=0.2.
\end{align*}
We assume that the system order $n=4$ is known, i.e., $\nu=4$.
From an open-loop simulation, a trajectory $\{u_k,y_k\}_{k=0}^{N-1}$ of length $N=1000$ is collected, where the output is subject to multiplicative measurement noise with signal-to-noise ratio $5\%$.
Problem~\eqref{eq:DD_sim_lsq_reg_lifted} with a squared exponential kernel with $\sigma=1$ is used to compute the output $\bar{y}$ resulting from a uniformly distributed random input $\bar{u}$ in the interval $[-0.3,0.3]$ of length $L=50$ with zero initial conditions.
The regularization parameter is chosen as $\lambda=10$.
Figure~\ref{fig:DD_sim_kernel_sin} shows the resulting output estimate as well as the true output for comparison.
It can be seen that the estimate is good, considering the noise level.
If the regularization term is omitted, i.e., $\lambda=0$, or a fixed number of polynomial basis functions is chosen, then the estimation accuracy deteriorates significantly, even for smaller noise levels.

\begin{figure}[h!]
\begin{center}
\includegraphics[width=0.49\textwidth]{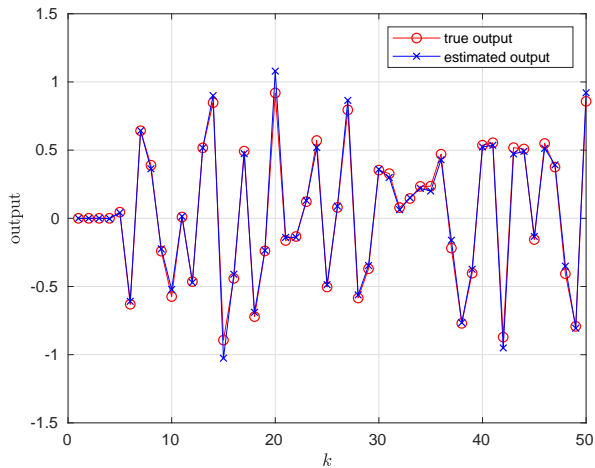}
\end{center}
\caption{True output and estimated output, computed via the proposed kernel-based data-driven simulation approach, corresponding to Example~\ref{ex:example}.}
\label{fig:DD_sim_kernel_sin}
\end{figure}
\end{example}

\section{Conclusion}\label{sec:conclusion}
This paper described a purely data-driven framework for system analysis and control.
All trajectories of an unknown system can be constructed from a single measured trajectory and thus, this trajectory captures all the required information needed for analysis and controller design, without explicit identification of a model.
After describing this result in the classical control framework, we extended it to certain classes of nonlinear systems and we applied this extension to the data-driven simulation problem via kernel methods.
Future research should further explore applications of the nonlinear extension presented in Section~\ref{sec:nonlinear_systems} to data-driven system analysis and control problems, as well as connections to more elaborate results from the literature on kernel methods.
 
\bibliographystyle{IEEEtran}  
\bibliography{Literature}

\end{document}